\documentclass[11pt]{amsart}
\usepackage{amscd,amssymb,verbatim,mathrsfs,color,bbold,amsmath,comment}
\newtheorem{theorem}{Theorem}
\newtheorem{corollary}[theorem]{Corollary}
\newtheorem{lemma}[theorem]{Lemma}
\newtheorem{proposition}[theorem]{Proposition}

\theoremstyle{definition}
\newtheorem{example}[theorem]{Example}
\newtheorem{remark}[theorem]{Remark}

\numberwithin{equation}{section}

\newcommand{\ep}{\varepsilon}

\newcommand{\abs}[1]{\lvert#1\rvert}

\newcommand{\bignorm}[1]{\bigl\lVert#1\bigr\rVert}
\newcommand{\Bigabs}[1]{\Bigl\lvert#1\Bigr\rvert}
\newcommand{\Bignorm}[1]{\Bigl\lVert#1\Bigr\rVert}

\newcommand{\N}{{\mathbb N}}
\newcommand{\R}{{\mathbb R}}

\newcommand{\E}{{\mathbb E}}

\newcommand{\cX}{{\mathcal X}}

\newcommand{\cC}{{\mathcal C}}
\newcommand{\cF}{{\mathcal F}}

\newcommand{\bP}{{\mathbb P}}

\def\one{\mathbb 1}

\title{On the extension property of dilatation monotone risk measures}

\author[M.~Rahsepar]{Massoomeh Rahsepar}
\address{Department of Mathematics, Ryerson University,  Toronto, Canada}
\email{srahnema@ryerson.ca}

\author[F.~Xanthos]{Foivos Xanthos}
\address{Department of Mathematics, Ryerson University,  Toronto, Canada}
\email{foivos@ryerson.ca}

\date{\today}

\thanks{This research was supported by an NSERC Discovery Grant. This first author is also supported by an Ontario Graduate Scholarship.}



\begin{document}

\begin{abstract}

Let $\cX$ be a subset of $L^1$ that contains the space of simple random variables $\mathcal{L}$ and $\rho: \cX \rightarrow (-\infty,\infty]$ a dilatation monotone functional with the Fatou property. In this note, we show that $\rho$ extends uniquely  to a $\sigma(L^1,\mathcal{L})$ lower semicontinuous and dilatation monotone functional $\overline{\rho}: L^1 \rightarrow (-\infty,\infty]$. Moreover,  $\overline{\rho}$ preserves monotonicity, (quasi)convexity, and cash-additivity of $\rho$. Our findings complement recent extension results for quasiconvex law-invariant functionals proved in \cite{LS:19,LT:19}. As an application of our results, we show that transformed norm risk measures on Orlicz hearts admit a natural extension to $L^1$ that retains the robust representations obtained in \cite{CL:09,KPR:10}.


\end{abstract}
\maketitle

\section{Introduction and Notations}

In the axiomatic theory of risk measures, it is customary to assume that risk measures are defined on a subspace $\cX$ of $L^1$. The problem of extending the domain of risk measures to the entire $L^1$ arises naturally in applications and has received significant attention in the mathematical finance literature (see e.g. \cite{CGX:18,FS:12,GLMX:18,KM:14,LS:19,LT:19}). The celebrated work \cite{FS:12}, asserts that if $\cX$ is some $L^p$-space for $1\leq p\leq \infty$, then any law-invariant convex risk measure $\rho:L^p \rightarrow (-\infty,\infty]$ with the Fatou property can be extended uniquely  to a law-invariant and $||\cdot||_1$ lower semicontinuous convex risk measure $\overline{\rho}:L^1 \rightarrow (-\infty,\infty]$. This result has been recently generalized to cover the case where $\cX$ is an Orlicz space or a more general rearrangement invariant space (\cite{CGX:18,GLMX:18,LS:19,LT:19}). A crucial step in the proof of the aforementioned results is that the underlying assumptions force the risk measure to be dilatation monotone. This property was introduced in \cite{L:04} and reflects the general belief that balancing out should never increase the involved risk. For a connection with other  notions of risk aversion, we refer the reader to \cite{CG:07,LS:19,MW:18,S:14}. 

In this paper, we study the extension problem for dilatation monotone risk measures  defined on a subset $\cX$ of $L^1$. The sole assumption we impose on the domain is that $\cX$ contains the space of simple random variables $\mathcal{L}$. This broad framework is  motivated by the  general class of transformed norm risk measures (\cite{CL:08,CL:09}), as in this case the domain typically is not a vector space (see Example~\ref{hgexmp}). In Section 2, we gather  the main results of the paper. Lemma ~\ref{Mlemma1} contains a result on convergence of finite conditional expectations that is of independent interest. In Proposition \ref{starproperty}, we show that the dilatation monotonicity property coupled with the Fatou property implies that the underlying risk measure is lower semicontinuous with respect to the relative $\sigma(L^1,\mathcal{L})$ topology. Theorem ~\ref{main_thm2} is our main extension result. In Corollary~\ref{lawcor}, we enlarge the class of domains $\cX$ for which the extension problem for quasiconvex law-invariant risk measures has a positive solution. Section 3 contains our application to transformed norm risk measures $T$. In Theorem \ref{tnormcor}, we calculate the extension $\overline{T}:L^1 \rightarrow (-\infty,\infty]$ and show that $\overline{T}$ preserves the desirable properties of $T$. In Corollary ~\ref{tnrom_cor}, we extend the dual representation of $T$ obtained in \cite{CL:09} and in Corollary~\ref{highorder_cor}, we extend the Kusuoka representation for higher order dual risk measures obtained in \cite{KPR:10}.




\subsection*{Notations} Let $(\Omega,\mathcal{F},\mathbb{P})$ be a fixed nonatomic probability space. We denote by $L^1$ the space of integrable random variables on $(\Omega,\mathcal{F},\mathbb{P})$ modulo almost sure equality and by $||\cdot||_1$ the $L^1$-norm. Throughout this paper, all equalities and inequalities are understood in the $\mathbb{P}$-almost sure (a.s.) sense.  The space of bounded random variables is denoted by $L^\infty$. A subspace $\cX$ of $L^1$ is said to be an ideal of $L^1$ whenever $|X| \leq |Y|$ and $Y \in \mathcal{X}$ imply $X \in \mathcal{X}$ for all $X \in L^1$.  For any subset $\mathcal{S}$ of $L^1$, we denote by $\mathcal{I}(\mathcal{S})$ the smallest ideal of $L^1$ including $\mathcal{S}$, which is given by  $\mathcal{I}(\mathcal{S})=\{X \in L^1 : n\in\N, \text{there exist }X_1,\dots,X_n \in \mathcal{S} \text{ and } \lambda_i \in \mathbb{R}_+ \text{ such that  } |X| \leq  \sum_{i=1}^n \lambda_i |X_i|\}$ (see e.g.~\cite{AB:06}, p.~322). For any $A \in \mathcal{F}$, we denote the indicator function by $\one_{A}$. A random variable is simple if it has the form $\sum_{i=1}^n a_i \one_{A_i}$, where $n \in \mathbb{N}, a_i \in \mathbb{R}$ and $A_i \in \mathcal{F}$. The space of simple random variables is denoted by $\mathcal{L}$.  We write $\pi$ to denote a finite measurable partition of $\Omega$ whose members have non-zero probabilities, and write $\Pi$ for the collection of all such $\pi$. We denote by $\sigma(\pi)$, the finite $\sigma$-subalgebra generated by $\pi$, and write $\mathbb{E}[X|\pi]:=\mathbb{E}[X|\sigma(\pi)]$ to denote the conditional expectation of $X \in L^1$ with respect to $\sigma(\pi)$. We recall below a few fundamental facts about  conditional expectations. For any $\pi=\{A_1,...,A_k\} \in \Pi$ and $X \in L^1$,  we have that $\mathbb{E}[X|\pi]=\sum_{i=1}^k \frac{\mathbb{E}[X\one_{A_i}]}{\mathbb{P}(A_i)}\one_{A_i}$ and  $||\mathbb{E}[X|\pi]||_1 \leq ||X||_1$. Moreover, $\mathbb{E}[X_n|\pi] \xrightarrow{{||\cdot||_1}} \mathbb{E}[X|\pi]$  for every sequence $(X_n)_{n \in \mathbb{N}} \subset L^1$ such that $X_n \xrightarrow{{||\cdot||_1}} X$.


\section{The main results}

\emph{In the following, $\cX$ will denote a subset of $L^1$ that contains the space of simple random variables $\mathcal{L}$ and $\rho: \mathcal{X} \rightarrow [-\infty,\infty]$}. We say that $\rho$  is   \emph{dilatation monotone} whenever $\rho(\mathbb{E}[X|\pi])\leq \rho(X)$ for any $X \in \mathcal{X}$ and $ \pi \in \Pi$ and that $\rho$ is \emph{law invariant} if for any $X,Y \in \mathcal{X}$ such that $X,Y$ have the same distribution under $\mathbb{P}$,   $\rho(X)=\rho(Y)$. $\rho$ is \emph{convex} (resp., \emph{quasiconvex}) whenever $\rho(\lambda X+(1-\lambda)Y)\leq \lambda \rho(X)+(1-\lambda) \rho(Y)$ (resp., $\rho(\lambda X+(1-\lambda)Y)\leq \max\{\rho(X),\rho(Y)\}$) for any $X,Y \in \mathcal{X}$ and $\lambda \in (0,1)$ such that $\lambda X+(1-\lambda) Y \in \mathcal{X}$. $\rho$ is \emph{cash-additive} whenever $\rho(X+s)=\rho(X)-s$ for any $X \in \mathcal{X}$ and $s \in \mathbb{R}$ with $X+s \in \mathcal{X}$. $\rho$ is \emph{monotone} whenever $\rho(X) \leq \rho(Y)$ for any $X,Y \in \mathcal{X}$ with $X \geq Y$. We say that $\rho$ has the  \emph{Fatou property} if for any sequence $(X_n)_{n \in \mathbb{N}}$ in $\mathcal{X}$ and $X \in \mathcal{X}$ we have that 
\[X_n \xrightarrow{{a.s.}} X \text{ and } \sup_{n \in \mathbb{N}}|X_n| \in \mathcal{I}(\mathcal{X}) \implies \rho(X) \leq \liminf_n \rho(X_n).   \tag{$\star$}\]
Let $\tau$ be a topology on $\cX$. We say that  $\rho$ is \emph{$\tau$ lower semicontinuous}, whenever the sublevel set $\{\rho\leq \lambda\}:=\{X \in \cX :\rho(X) \leq \lambda\}$ is $\tau$ closed for each $\lambda \in \mathbb{R}$. We also  recall that if $\tau$ is a metrizable topology, then $\rho$ is $\tau$ lower semicontinuous if and only if for any sequence $(X_n)_{n \in \mathbb{N}}$ in $\mathcal{X}$ and $X \in \mathcal{X}$ we have that $X_n \xrightarrow{{\tau}} X$ implies $\rho(X) \leq \liminf_n \rho(X_n)$ (see, e.g., (\cite{AB:06}, Lemma 2.42)).

\begin{remark}
We note here that in the standard framework where $\cX$ is a rearrangement invariant space, or more generally, an ideal of $L^1$, the supremum taken in ($\star$) is an element of $\mathcal{X}$ and thus the definition of the Fatou property we give here coincides with the one  used in the literature (see e.g. \cite{GX:18} and the references therein). 
\end{remark}

Convergence of finite conditional expectations plays a key role in the analysis of dilatation monotone functionals. It is a well-known fact that for any $X \in L^1$, one can construct a sequence $(\pi_n)_{n \in \mathbb{N}} \subset \Pi$ such that $\mathbb{E}[X|\pi_n] \xrightarrow{{a.s.}} X$ and $\sup_{n \in \mathbb{N}}|\mathbb{E}[X|\pi_n]| \in L^1$. The following result tells us that we can pick $(\pi_n)_{n \in \mathbb{N}}$ in such a way that $\sup_{n \in \mathbb{N}}|\mathbb{E}[X|\pi_n]| \in \mathcal{I}(\{X,\one\})$, which enables us to give a positive answer to Question 2.8.2 in \cite{CGX:18}.

\begin{lemma}\label{Mlemma1}
For any $X\in L^1$, there exists $(\pi_n)_{n \in \mathbb{N}} \subset \Pi$ and $k \in \mathbb{R}_+$ such that $$\mathbb{E}[X|\pi_n] \xrightarrow{a.s.} X  \text{ and } |\mathbb{E}[X|\pi_n]| \leq |X|+k, \quad\text{ for all } n\in \mathbb{N}$$
\end{lemma}
\begin{proof}
Let $X\in L^1$ and  $k_1 \in \mathbb{R}_+$ be such that  $\mathbb{P}(|X|\leq k_1)>\frac{1}{2}$. Pick any $0<\epsilon<\frac{1}{2}$. Since $(\Omega,\mathcal{F},\mathbb{P})$ is nonatomic, we can find $A \in \mathcal{F}$  such that $A\subset \{|X|\leq k_1\}$ and $\mathbb{P}(A)=\epsilon$. 
In addition, since $X \in L^1$, we can find $k_2 \in \mathbb{R}$ such that $k_2>k_1$ and $\E[|X|\one_{|X|>k_2}]<\epsilon$. Put $\Omega'=\{|X| \leq k_2\}\setminus A$ and note that $\mathbb{P}(\Omega')=\mathbb{P}(|X| \leq k_2)-\mathbb{P}(A)\geq \mathbb{P}(|X|\leq k_1)-\mathbb{P}(A)>0$. For the localized probability space $(\Omega',\mathcal{F}_{|\Omega'},\mathbb{P}_{|\Omega'})$, where $\mathcal{F}_{|\Omega'}=\{E\in\cF:E\subset\Omega'\}$ and $\mathbb{P}_{|\Omega'}(E)=\frac{\bP(E)}{\bP(\Omega')}$, we have that $X\one_{\Omega'} \in L^\infty(\Omega',\mathcal{F}_{|\Omega'},\mathbb{P}_{|\Omega'})$. Applying (\cite{GLMX:18},  Lemma 3.1) to $L^\infty(\Omega',\mathcal{F}_{|\Omega'},\mathbb{P}_{|\Omega'})$, we can find a measurable partition $\{B_1,B_2,...,B_n\}$ of $\Omega'$ such that $\mathbb{P}(B_i)>0$ for each $i$ and  
\begin{equation}\label{Bounded3}
    \Bigabs{\sum_{i=1}^n \frac{\E[X\one_{B_i}]}{\mathbb{P}(B_i)}\one_{B_i}-X\one_{\Omega'}}<\epsilon \one_{\Omega'}.
\end{equation}

Note that $\mathbb{P}(\Omega\setminus \Omega')\geq \mathbb{P}(A)=\epsilon>0$. Thus $\pi=\{B_1,B_2,\dots, B_n,\Omega\setminus \Omega'\} \in \Pi$. Applying (\ref{Bounded3}), we obtain 
\begin{align}\label{Bounded1}
      |\mathbb{E}[X|\pi]\one_{\Omega'}|\leq&|\mathbb{E}[X|\pi]\one_{\Omega'}-X\one_{\Omega'}|+|X\one_{\Omega'}|=\Bigabs{\sum_{i=1}^n \frac{\E[X\one_{B_i}]}{\mathbb{P}(B_i)}\one_{B_i}-X\one_{\Omega'}}+|X\one_{\Omega'}|\\
\nonumber      <&\epsilon\one_{\Omega'}+|X|<\frac{1}{2}\one_{\Omega'}+|X|.
\end{align}
Now, we will find an upper bound for $|\mathbb{E}[X|\pi]\one_{\Omega \setminus \Omega'}|$. Note that 
\begin{equation}\label{Bounded0}
  \E[|X|\one_{\Omega \setminus \Omega'}]=\E[|X|\one_{A \cup\{ |X|>k_2\}}]=\E[|X|\one_{A}]+\E[|X|\one_{\{|X|>k_2\}}] <\epsilon(k_1+1).
\end{equation}
In view of  $\mathbb{P}(\Omega\setminus \Omega') \geq \epsilon$, we get that 
$$|\mathbb{E}[X|\pi]\one_{\Omega \setminus \Omega'}|\leq \frac{\mathbb{E}[|X| \one_{\Omega \setminus \Omega'}]}{\mathbb{P}(\Omega\setminus \Omega')}\one_{\Omega \setminus \Omega'} <(k_1+1)\one_{\Omega \setminus \Omega'}.$$
Therefore, by invoking ~(\ref{Bounded1}), it follows that
 \begin{equation}\label{Bounded5}
     |\mathbb{E}[X|\pi]|\leq |\mathbb{E}[X|\pi] \one_{\Omega'}| + |\mathbb{E}[X|\pi]\one_{\Omega\setminus \Omega'}|\leq |X|+k_1+1.
 \end{equation}
We next claim that 
\begin{equation}\label{Bounded6}
||\mathbb{E}[X|\pi]-X||_1 < \epsilon(3+2k_1).
\end{equation}
Indeed, by applying (\ref{Bounded3}) and (\ref{Bounded0}),  we get the following 
\begin{align*}
    ||\mathbb{E}[X|\pi]-X||_1&=||(\mathbb{E}[X|\pi]-X)\one_{\Omega'}||_1+||(\mathbb{E}[X|\pi]-X)\one_{\Omega \setminus \Omega'}||_1\\
    &\leq \epsilon+||\mathbb{E}[X|\pi]\one_{\Omega \setminus \Omega'}||_1+||X\one_{\Omega \setminus \Omega'}||_1=\epsilon+||\mathbb{E}[X\one_{\Omega \setminus \Omega'}|\pi]||_1+||X\one_{\Omega \setminus \Omega'}||_1\\
    & \leq \epsilon+2||X\one_{\Omega \setminus \Omega'}||_1 < \epsilon(3+2k_1).
\end{align*}

Now letting, e.g., $\ep=\frac{1}{n}$ in (\ref{Bounded5}) and (\ref{Bounded6}), we obtain $(\pi_n)_{n\in \N}\subset \Pi$ such that 
\begin{equation*}
    \sup_{n\in\N}|\mathbb{E}[X|\pi_{n}|\leq  |X|+(k_1+1)\one \quad\text{ and } \quad \mathbb{E}[X|\pi_n] \xrightarrow{{|| \cdot ||_1}} X. 
\end{equation*}
By passing to a subsequence, we may assume that  $\mathbb{E}[X|\pi_{n}]\xrightarrow{{a.s}} X$. 
\end{proof}


In general, the Fatou property is a weaker property than relative $||\cdot||_1$ lower semicontinuity. In particular, for $\cX=L^p,1\leq p<+\infty$, the Fatou property is equivalent to $|| \cdot||_p$ lower semicontinuity (see e.g. \cite{MR:09}). In the following proposition, we show that if $\rho$ is dilatation monotone, then the Fatou property is, in fact, equivalent to lower semicontinuity with respect to the  relative $\sigma(L^1,\mathcal{L})$ topology on $\cX$.

\begin{proposition}\label{starproperty}
Let $\rho:\cX \rightarrow [-\infty,\infty]$ be dilatation monotone. Then the following are equivalent
\begin{enumerate}
\item[(i)] $\rho$ has the  Fatou property. 
\item[(ii)] $\rho$ is lower semicontinuous with respect to the relative $\sigma(L^1,\mathcal{L})$ topology on $\cX$.
\item[(iii)] $\rho$ is lower semicontinuous with respect to the relative $||\cdot||_1$ topology on $\cX$.
\end{enumerate}
If any of the above holds, then for any $X \in \mathcal{X}$,
\begin{equation}\label{remark00011}
\rho(X)=\sup_{\pi\in\Pi}\rho(\E[X|\pi]),\end{equation}
and for any   $(\pi_n)_{n \in \mathbb{N}} \subset \Pi$, 
\begin{equation}\label{remark000}
\mathbb{E}[X|\pi_n] \xrightarrow{{\sigma(L^1,\mathcal{L})}} X \implies \rho(X)=\lim_n\rho(\mathbb{E}[X|\pi_n]) .\end{equation}
\end{proposition}

\begin{proof}
$(i) \Rightarrow (ii). $ Let $\{\rho \leq \lambda\}$ be a sublevel set and  $(X_{\alpha}) \subset \{\rho \leq \lambda\}$ be any net such that $X_{\alpha} \xrightarrow{{\sigma(L^1,\mathcal{L})}} X \in \mathcal{X}$.  By Lemma~\ref{Mlemma1}, we can find $k \in \mathbb{R}_+$ and $(\pi_m)_{m \in \mathbb{N}}  \subset \Pi$ such that $\mathbb{E}[X|\pi_m]\xrightarrow{{a.s.}} X$ and $\sup_{m \in \mathbb{N}}|\mathbb{E}[X|\pi_m]| \leq |X|+k$. By ($\star$), we have that 
\begin{equation}\label{eq201}\rho(X) \leq \liminf_m \rho(\mathbb{E}[X|\pi_m]).\end{equation}

Note that for any $\pi=\{A_1,...,A_k\} \in \Pi$, $\mathbb{E}[X_{\alpha}|\pi] \xrightarrow{{||\cdot||_1}} \mathbb{E}[X|\pi]$. Indeed, we have that 
\begin{align*}
    \Bignorm{\mathbb{E}[X_{\alpha}|\pi]-\mathbb{E}[X|\pi]}_1=&\bignorm{\mathbb{E}[X-X_\alpha|\pi]}_1
=\Bignorm{\sum_{i=1}^k\frac{\mathbb{E}[X_{\alpha}\one_{A_i}-X\one_{A_i}]}{\mathbb{P}(A_i)}\one_{A_i}}_1\\
\leq &\mathbb{E}\Big[\sum_{i=1}^k\Bigabs{\frac{\mathbb{E}[X_{\alpha}\one_{A_i}-X\one_{A_i}]}{\mathbb{P}(A_i)}}\one_{A_i}\Big]=\sum_{i=1}^k |\mathbb{E}[X_a\one_{A_i}-X\one_{A_i}]|\rightarrow 0 .
\end{align*}
Thus, we can choose $(\alpha_n)_{n \in \mathbb{N}}$ such that $||\mathbb{E}[X_{\alpha_n}|\pi]-\mathbb{E}[X|\pi]||_1\leq \frac{1}{2^n}$ for all $n$. Then  $Y:=\sum_{k=1}^\infty|\mathbb{E}[X_{\alpha_k}-X|\pi]|\in L^1$. Note that for any $n \in \mathbb{N}$, $\sum_{k=1}^n|\mathbb{E}[X_{\alpha_k}-X|\pi]| \in \mathcal{R}_{\pi}$, the space of random variables measurable with respect to $\sigma(\pi)$. Since $\mathcal{R}_{\pi}$ is a finite dimensional subspace of $L^1$, it follows that $\mathcal{R}_{\pi}$ is a closed subspace of $L^1$ and thus   $Y \in \mathcal{R}_{\pi}$. Moreover,
$$|\mathbb{E}[X_{\alpha_n}|\pi]-\mathbb{E}[X|\pi]|\leq \sum_{k=n}^\infty |\mathbb{E}[X_{\alpha_k}|\pi]-\mathbb{E}[X|\pi]| \xrightarrow{{a.s.}} 0,$$
$$\sup_{n \in \mathbb{N}}\{|\mathbb{E}[X_{\alpha_n}|\pi]|\}\leq Y+|\mathbb{E}[X|\pi]| \in \mathcal{L}\subset  \cX \subset \mathcal{I}(\mathcal{X}).$$
Therefore, by applying again ($\star$) to $\pi_m$, we get that, for any $m\in\N$,
$$\rho(\mathbb{E}[X|\pi_m]) \leq \liminf_n\rho(\mathbb{E}[X_{\alpha_n}|\pi_m]) .$$
Applying  dilatation monotonicity of $\rho$ to the right hand side, we have that, for any $m\in\N$,
$$\rho(\mathbb{E}[X|\pi_m])\leq \liminf_n \rho(X_{\alpha_n}).$$
Thus,  by ~(\ref{eq201}), we get that $\rho(X) \leq \liminf_n \rho(X_{\alpha_n})\leq \lambda$ and thus $X \in \{\rho \leq \lambda\}$.

$(ii) \Rightarrow (iii)$. This is immediate since the $\sigma(L^1,\mathcal{L})$ topology is weaker than the $||\cdot||_1$ topology.

$(iii) \Rightarrow (i).$ Let $(X_n)_{n \in \mathbb{N}} \subset \cX$ and $X \in \cX$ be such that $X_n \xrightarrow{{a.s.}} X$ and $ \sup_{n \in \mathbb{N}}|X_n| \in \mathcal{I}(\mathcal{X}) $. By the Dominated Convergence Theorem, we have that $X_n \xrightarrow{{||\cdot||_1}} X$ and (iii) yields that $\rho(X)\leq \liminf_n \rho(X_n)$.

To verify~(\ref{remark000}), let $\mathbb{E}[X|\pi_n] \xrightarrow{\sigma(L^1,\mathcal{L})} X$ and note that by (ii) and the dilatation monotonicity  of $\rho$, we have that  $\rho(X) \leq \liminf_n \rho(\mathbb{E}[X|\pi_n]) \leq \limsup_n \rho(\mathbb{E}[X|\pi_n]) \leq \rho(X)$. In particular, it follows that $\rho(X)=\lim_n\rho(\mathbb{E}[X|\pi_n])$. (\ref{remark00011}) follows by the dilatation monotonicity  of $\rho$, (i) and Lemma~\ref{Mlemma1}.
\end{proof}

Now, we turn to the extension problem. We are motivated by (\ref{remark00011}) to define the following functional associated with $\rho$.
\begin{equation}\label{eq1}\overline{\rho}(X)=\sup_{\pi\in\Pi}\rho(\E[X|\pi]),\quad X \in L^1.\end{equation}

\begin{theorem}\label{main_thm2}
Let $\rho: \mathcal{X} \rightarrow (-\infty,\infty]$ be a dilatation monotone functional with the Fatou property. Then $\overline{\rho}: L^1 \rightarrow (-\infty,\infty]$ is the unique extension of $\rho$ to $L^1$ that is dilatation monotone and  $\sigma(L^1,\mathcal{L})$ lower semicontinuous. Moreover, $\overline{\rho}$ is law-invariant and preserves monotonicity, (quasi)convexity, and cash-additivity of $\rho$.
\end{theorem}

\begin{proof}In view of (\ref{remark00011}), it is immediate to see that $\overline{\rho}$ agrees with $\rho$  on $\cX$. Thus for any $X \in L^1$ and $\pi \in \Pi$, since $\E[X|\pi]\in \mathcal{L}\subset \cX$, we get that $$\overline{\rho}(\mathbb{E}[X|\pi])=\rho(\mathbb{E}[X|\pi])\leq \overline{\rho}(X).$$
This establishes  dilatation monotonicity of $\overline{\rho}$ on $L^1$. Since $\rho>-\infty$ on $\cX$, it also follows that $\overline{\rho}>-\infty $ on $L^1$. 
Next, we  show that $\overline{\rho}$ is $||\cdot||_1$ lower semicontinuous, and thus by Proposition~\ref{starproperty}, we get that $\overline{\rho}$ is $\sigma(L^1,\mathcal{L})$ lower semicontinuous. Indeed, let $X_n \xrightarrow{|| \cdot ||_1} X$ and fix some $\pi \in \Pi$. Then $\mathbb{E}[X_n|\pi] \xrightarrow{|| \cdot ||_1} \mathbb{E}[X|\pi]$. Therefore, by the Fatou property of $\rho$ and the definition of $\overline{\rho}$,
$$\rho(\mathbb{E}[X|\pi]) \leq \liminf_n \rho(\mathbb{E}[X_n|\pi])\leq \liminf_n \overline{\rho}(X_n).$$
Taking supremum over $\pi$, we get that $\overline{\rho}(X) \leq \liminf_n \overline{\rho}(X_n)$ and thus $\overline{\rho}$ is $||\cdot||_1$ lower semicontinuous. The uniqueness of $\overline{\rho}$ is immediate by ~(\ref{remark00011}). 

The law-invariance of $\overline{\rho}$ follows by (\cite{LS:19}, Theorem 18(ii)).
Assume that $\rho$ is quasiconvex. Pick any $X_1,X_2 \in L^1$ and $ \lambda \in (0,1)$. Then \begin{align*}
\overline{\rho} (\lambda X_1+(1-\lambda) X_2)&=\sup_{\pi\in\Pi}\rho(\E[\lambda X_1+(1-\lambda) X_2|\pi])=\sup_{\pi\in\Pi}\rho(\lambda \E[X_1|\pi]+(1-\lambda)\E[ X_2|\pi])\\
\leq &\sup_{\pi\in\Pi}\max\{\rho(\mathbb{E}[X_1|\pi]),\rho(\mathbb{E}[X_2|\pi)\}\Big) 
\leq \max\{\overline{\rho}(X_1),\overline{\rho}(X_2)\}. 
\end{align*}
Hence, $\overline{\rho}$ is quasiconvex as well.
In a similar manner, one shows that $\overline{\rho}$ preserves convexity, monotonicity and cash-additivity of $\rho$.
\end{proof}

To present our applications to the study of quasiconvex law-invariant functionals, we need to apply the following result from \cite{CGX:18}. In particular, the proof of (\cite{CGX:18}, Proposition 2.2) yields the following construction, which we isolate here.

\begin{lemma}(\cite{CGX:18})\label{CGX:18prop}
For any $X\in L^1$ and $\pi \in \Pi$, there exists $k \in \mathbb{R}_+$ and  $X_{n,j}, j=1,...,N_n,n \in \mathbb{N}$ such that $X_{n,j}\one_{\{|X|\leq n\}}$ has the same distribution as $X \one_{\{|X|\leq n\}},X_{n,j}\one_{\{|X|>n\}}=X\one_{\{|X|>n\}}$ for each $j,n$ and  $$\frac{1}{N_n}\sum_{j=1}^{N_{n}} X_{n,j} \xrightarrow{a.s.} \mathbb{E}[X|\pi]  \text{ and } |\frac{1}{N_n}\sum_{j=1}^{N_{n}} X_{n,j}| \leq |X|+|\mathbb{E}[X|\pi]|+k, \quad\text{ for all } n \in \mathbb{N}.$$
\end{lemma}

\begin{corollary}\label{lawcor}
Let $\cX$ be a convex subset of $L^1$ with the following properties

\begin{itemize}
\item[(i)] $\cX+\cX \subset \cX$,
\item[(ii)] $L^\infty \subset \cX$ and $X \one_{A} \in \cX$ for each $X \in \cX$ and $A \in \mathcal{F}$.
\end{itemize}

Then for any quasiconvex law-invariant $\rho:\cX \rightarrow (-\infty,\infty]$ with the Fatou property, we have that $\overline{\rho}:L^1 \rightarrow (-\infty,\infty]$, defined in (\ref{eq1}), agrees with $\rho$ on $\cX$, is quasiconvex, law-invariant, dilatation monotone and $\sigma(L^1,\mathcal{L})$ lower semicontinuous. Moreover, $\overline{\rho}$ preserves monotonicity, convexity, and cash-additivity of $\rho$.
\end{corollary}

\begin{proof}
We will verify that $\rho$ is dilatation monotone, then the result follows by Theorem ~\ref{main_thm2}. Let $X \in \mathcal{X}$  and $\pi \in \Pi$. Pick $X_{n,j}, j=1,...,N_n,n \in \mathbb{N}$ as in Lemma~\ref{CGX:18prop}. Then  $X_{n,j}\one_{\{|X|\leq n\}} \in L^\infty$, and thus $X_{n,j}=X_{n,j}\one_{\{|X|\leq n\}}+X_{n,j}\one_{\{|X|>n\}}=X_{n,j}\one_{\{|X|\leq n\}}+X\one_{\{|X|>n\}} \in \cX$. We also have that $X_{n,j}$ has the same distribution as $X$, therefore $\rho(X)=\rho(X_{n,j})$ for each $j,n$ and by the quasi-convexity of $\rho$ we get that  $\rho(\frac{1}{N_n}\sum_{j=1}^{N_{n}} X_{n,j})\leq \rho(X)$. Finally by ($\star$), it follows that $\rho(\mathbb{E}[X|\pi]) \leq \liminf_n \rho(\frac{1}{N_n}\sum_{j=1}^{N_{n}} X_{n,j})\leq \rho(X)$.
\end{proof}

Theorem~\ref{main_thm2} and Corollary ~\ref{lawcor} improve and extend some recent results in this topic (see e.g. (\cite{LT:19}, Theorem 3.1) and (\cite{LS:19}, Theorem 18, Proposition 20))  to domains $\cX$ that are not necessarily linear nor rearrangement invariant. This general framework besides being of mathematical interest is also motivated by the Musielak-Orlicz spaces which are relevant to utility theory (see e.g. \cite{FM:11}).

In the following, we will derive a dual representation for convex $\rho$ in terms of the following conjugate function.
\begin{equation}\label{eq2}\rho^{\#}(Y)=\sup_{X \in \mathcal{L}}\{\mathbb{E}[XY]-\rho(X)\},\qquad Y \in L^1.\end{equation}
We say that a subset $\cC$ of $L^1$ is  a \emph{dilatation monotone set} if it is non-empty and $\mathbb{E}[X|\pi] \in \cC$ for all $X \in \cC$ and $\pi \in \Pi$.

\begin{lemma}\label{density}
Let $\rho:\mathcal{X} \rightarrow [-\infty,\infty]$ be dilatation monotone. 
\begin{itemize}
\item[(i)]  $\rho^{\#}$ is dilatation monotone and $||\cdot||_1$ lower semicontinuous.
\item[(ii)] Let $\cC \subset \cX$ be a dilatation subset of $L^1$. If $\rho$ has the Fatou property, then for any  $Y \in L^1$ such that $\mathbb{E}[|YX|] <\infty$ for all $X \in \mathcal{C}$, we have that 
$$\sup_{X \in \mathcal{L} \cap \cC}\{\mathbb{E}[XY]-\rho(X)\}=\sup_{X \in \mathcal{C}}\{\mathbb{E}[XY]-\rho(X)\}.$$
\end{itemize}
\end{lemma}

\begin{proof}
(i) Fix any $Y \in L^1$ and $\pi \in \Pi$. We have that
\begin{align*}
\rho^{\#}(\mathbb{E}[Y|\pi])=&\sup_{X \in \mathcal{L}}\{\mathbb{E}[X\mathbb{E}[Y|\pi]]-\rho(X)\}=\sup_{X \in \mathcal{L}}\{\mathbb{E}[Y\mathbb{E}[X|\pi]]-\rho(X)\}\\
    \leq &\sup_{X \in \mathcal{L}}\{\mathbb{E}[Y\mathbb{E}[X|\pi]]-\rho(\mathbb{E}[X|\pi])\}\leq\sup_{Z \in \mathcal{L}}\{\mathbb{E}[YZ]-\rho(Z)\}=\rho^{\#}(Y).
\end{align*}
Thus $\rho^{\#}$ is dilatation monotone on $L^1$. Next, let $(Y_n)_{n \in \mathbb{N}} \subset L^1$ and $Y \in L^1$ such that $Y_n \xrightarrow{{|| \cdot ||_1}} Y$. Then $\lim_n(\mathbb{E}[XY_n]-\rho(X))=\mathbb{E}[XY]-\rho(X)$ for all $X \in \mathcal{L}$. Thus since $\rho^{\#}(Y_n) \geq  \mathbb{E}[XY_n]-\rho(X)$ for all $ X \in \mathcal{L}$, it follows that $\liminf_n\rho^{\#}(Y_n) \geq  \mathbb{E}[XY]-\rho(X)$ for any $ X \in \mathcal{L}$, implying that $\liminf_n\rho^{\#}(Y_n) \geq \rho^{\#} (Y)$. This proves that $\rho^{\#}$ is $||\cdot||_1$ lower semicontinuous.

(ii) It is clear that $c:=\sup_{X \in \mathcal{L} \cap \cC}\{\mathbb{E}[XY]-\rho(X)\} \leq \sup_{X \in \mathcal{C}}\{\mathbb{E}[XY]-\rho(X)\}$. Let $Y \in L^1$ be such that $\mathbb{E}[|YX|] <\infty$ for all $X \in \mathcal{C}$. Now pick any $X \in \mathcal{C}$. By Lemma~\ref{Mlemma1}, there exists a sequence $(\pi_n)_{n \in \mathbb{N}} \subset \Pi$ such that $\mathbb{E}[X|\pi_n] \xrightarrow{a.s.} X$ and $X^*:=\sup_{n \in \mathbb{N}}\{|\mathbb{E}[X|\pi_n]|\} \in \mathcal{I}(\cC\cup\{\one\})$, so that $\E[X^*\abs{Y}]<\infty$. By the Dominated Convergence Theorem,   $\mathbb{E}[YX]=\lim_{n} \mathbb{E}[Y\mathbb{E}[X|\pi_n]]$. Thus since $ c\geq \mathbb{E}[\mathbb{E}[X|\pi_n]Y]-\rho(\mathbb{E}[X|\pi_n])$ for any $n\in\N$, by  ~(\ref{remark000}) we get that 
$$c\geq \lim_n \Big(\mathbb{E}[\mathbb{E}[X|\pi_n]Y]-\rho(\mathbb{E}[X|\pi_n])\Big)=\mathbb{E}[XY]-\rho(X).$$
Therefore, $c \geq \sup_{X \in \mathcal{C}}\{\mathbb{E}[XY]-\rho(X)\} $, yielding the desired inequality.
\end{proof}

\begin{theorem}\label{main_rep_thm}
Let $\rho:\cX \rightarrow (-\infty,\infty]$ be convex, dilatation monotone  with the Fatou property, then $\overline{\rho}$ admits the following dual representation 
$$\overline{\rho}(X)=\sup_{Y \in \mathcal{L}}\{\mathbb{E}[XY]-\rho^{\#}(Y)\}, \quad\text{ for all } X \in L^1 .$$
If, moreover, $\rho$ is not identically equal to $\infty$, is cash-additive and monotone, then the dual representation can be simplified as follows
$$\overline{\rho}(X)=\sup_{\frac{d\mathbb{Q}}{d\mathbb{P}} \in \mathcal{L}}\{\mathbb{E}_{\mathbb{Q}}[-X]-\rho^{\#}(-\frac{d\mathbb{Q}}{d\mathbb{P}})\},\quad\text{ for all } X \in L^1.$$
\end{theorem}

\begin{proof}
By Theorem~\ref{main_thm2}, $\overline{\rho}:L^1 \rightarrow (-\infty,\infty]$ is convex, dilatation monotone and $\sigma(L^1,\mathcal{L})$ lower semicontinuous. Let $\overline{\rho}^*$ be the convex conjugate of $\overline{\rho}$, that is,   $\overline{\rho}^{*}(Y)=\sup_{X \in L^1} \{\mathbb{E}[XY]-\overline{\rho}(X)\}$,  $Y \in \mathcal{L}$. By Lemma~\ref{density}(ii) applied with $\cC=L^1$, it follows that \begin{equation}\label{eqdual0}\overline{\rho}^{*}(Y)=\sup_{X \in \mathcal{L}} \{\mathbb{E}[XY]-\rho(X)\}=\rho^{\#}(Y),\end{equation} for all $Y \in \mathcal{L}$. 
By standard convex duality results (see, e.g., (\cite{ET:99}, Proposition 4.1)) and ~(\ref{eqdual0}), we get that 
\begin{equation}\label{eqdual}\overline{\rho}(X)=\sup_{Y \in \mathcal{L}} \{\mathbb{E}[XY]-\rho^{\#}(Y)\}.\end{equation}
In the case where $\rho$ is not identically equal to $\infty$, is cash-additive and monotone it is standard to check that  $\rho^{\#}(-Y)=\infty$ for all $Y \in \mathcal{L}\setminus \mathcal{D}$, where $\mathcal{D}=\{Y \in \mathcal{L}_+\,\, : \,\, \mathbb{E}[Y]=1\}$. Indeed, fix $X_0 \in \mathcal{L}$ such that $\rho(X_0) \in \mathbb{R}$ and let $Y \in \mathcal{L}$ such that $\mathbb{E}[Y] \neq 1$, then we have that $\rho^{\#}(-Y)\geq \mathbb{E}[-Y(X_0+m)]-\rho(X_0+m)=m(1-\mathbb{E}[Y])+\mathbb{E}[-YX_0]-\rho(X_0)$ for all $m \in \mathbb{R}$, thus $\rho^{\#}(-Y)=\infty$. Suppose that $Y \not\geq 0$, then $\mathbb{E}[-Y\one_{\{Y<0\}}]>0$ and we have that $\rho^{\#}(-Y)\geq \mathbb{E}[-Y(X_0+m\one_{\{Y<0\}})]-\rho(X_0+m\one_{\{Y<0\}})\geq m \mathbb{E}[-Y\one_{\{Y<0\}}]+\mathbb{E}[-YX_0]-\rho(X_0)$ for all $m \in \mathbb{R}_+$, thus $\rho^{\#}(-Y)=\infty$. Therefore the representation ~(\ref{eqdual}) takes the form: $\overline{\rho}(X)=\sup_{Y \in \mathcal{L}} \{\mathbb{E}[-XY]-\rho^{\#}(-Y)\}=\sup_{Y \in \mathcal{D}} \{\mathbb{E}[-XY]-\rho^{\#}(-Y)\}=\sup_{\frac{d\mathbb{Q}}{d\mathbb{P}} \in \mathcal{L}} \{\mathbb{E}_{\mathbb{Q}}[-X]-\overline{\rho}^{\#}(-\frac{d\mathbb{Q}}{d\mathbb{P}})\}$.


\end{proof}


We will now proceed to analyzing cash-additive hulls of dilatation monotone functionals. For $f:L^1 \rightarrow (-\infty,\infty]$, recall that  the cash-additive hull $\rho^f$ of $f$ introduced in \cite{FK:07} is defined as follows:
\begin{equation}\label{cash0}
\rho^f(X)=\inf_{s \in \mathbb{R}}\{f(X-s)-s\},\quad X \in L^1.
 \end{equation}
From the definition of $\rho^f$, it is clear that $\rho^f$ is cash-additive and is easy to see that $\rho^f$ preserves convexity, monotonicity and dilatation monotonicity of $f$. In the following result, we show that under a coercive condition, the infimum in ~(\ref{cash0}) is attained, and moreover, $\rho^f$ preserves continuity properties of $f$. 





\begin{theorem}\label{cash-thm}
Let $f:L^1 \rightarrow (-\infty,\infty]$ be dilatation monotone and $||\cdot||_{1}$ lower semicontinuous that satisfies the following coercive condition.
\begin{equation}\label{coercive}\displaystyle\lim_{|s| \rightarrow \infty} f(s)+s=\infty.\end{equation}
Then $\rho^f:L^1 \rightarrow (-\infty,\infty]$ is $|| \cdot||_1$ lower semicontinuous and 
$$
\rho^f(X)=\min_{s \in \mathbb{R}}\{f(X-s)-s\},\quad\text{ for all } X \in L^1.
$$
\end{theorem}

\begin{proof} Pick any $X \in L^1$. If $\rho^f(X)=\infty$, then  $f(X-s)-s=\infty$ for all $s \in \mathbb{R}$, and thus the infimum in (\ref{cash0}) is attained at every $s\in\R$. 
Now assume that $\rho^f(X)<\infty$.
Let $(s_n)_{n \in \mathbb{N}} \subset \mathbb{R}$ such that $f(X-s_n)-s_n \rightarrow \rho^f(X)$. Suppose that $(s_n)_{n \in \mathbb{N}}$ is unbounded. By passing to a subsequence, we may assume that $|s_n| \rightarrow \infty$.	Since $f$ is dilatation monotone, we have $$f(X-s_n)-s_n \geq f(\mathbb{E}[X]-s_n)+(\mathbb{E}[X]-s_n)-\mathbb{E}[X],$$
implying that
$$\rho^f(X) \geq  \limsup_n\Big(f(\mathbb{E}[X]-s_n)+(\mathbb{E}[X]-s_n)\Big)-\mathbb{E}[X].$$ 
Since $|\mathbb{E}[X]-s_n|\geq \abs{s_n}-|\mathbb{E}[X]| \rightarrow \infty$,   (\ref{coercive}) implies that $\rho^f(X)=\infty$, which is a contradiction. Thus $(s_n)_{n \in \mathbb{N}}$ is bounded, and by passing to a subsequence, we may assume that $s_n \rightarrow s \in \mathbb{R}$. Therefore, $X-s_n \xrightarrow{||\cdot||_1} X-s$. Since $f$ is $||\cdot||_1$ lower semicontinuous, we get that $f(X-s)-s\leq \liminf_n (f(X-s_n)-s_n)=\rho^f(X)$. In particular, it follows that $\rho^f(X)=f(X-s)-s$, i.e., the infimum in (\ref{cash0}) is attained, and consequently, $\rho^f(X)\neq -\infty$. 


It remains to be shown that $\rho^f$ is $||\cdot||_1$ lower semicontinuous. Fix some $\lambda \in \mathbb{R}$ and let $(X_n)_{n \in \mathbb{N}}\subset L^1,X\in L^1$ be such that $X_n\xrightarrow{||\cdot||_1} X$ and $(X_n)_{n \in \mathbb{N}} \subset \{\rho \leq \lambda\}$. Let $(s_n)_{n \in \mathbb{N}} \subset \mathbb{R}$ be such that $\rho^f(X_n)=f(X_n-s_n)-s_n$. Suppose that $(s_{n})_{n \in \mathbb{N}}$ is unbounded. Then we can find a subsequence $(s_{n_k})_{k \in \mathbb{N}}$ such that $|s_{n_k}| \rightarrow \infty$. Put $t_k=\mathbb{E}[X_{n_k}]-s_{n_k}$ and note that $\mathbb{E}[X_{n_k}] \rightarrow \mathbb{E}[X] \in \mathbb{R}$ and $|t_k| \geq |s_{n_k}|-|\mathbb{E}[X_{n_k}]|\rightarrow \infty$. Since $f$ is dilatation monotone, we have that $\rho^f(X_{n_k})=f(X_{n_k}-s_{n_k})-s_{n_k}\geq f(t_k )+t_k-\mathbb{E}[X_{n_k}]$. By ~(\ref{coercive}),   $f(t_k)+t_k \rightarrow \infty$, and thus  $\lim_k \rho^f(X_{n_k})=\infty$, which is a contradiction. Thus, $(s_n)_{n \in \mathbb{N}}$ is bounded, and we can extract a subsequence $(s_{n_k})_{k \in \mathbb{N}}$ such that $s_{n_k} \rightarrow s \in \mathbb{R}$. Then  $X_{n_k}-s_{n_k} \xrightarrow{||\cdot||_1} X-s$. Since $f$ is $||\cdot||_1$ lower semicontinuous, we have that $f(X-s)\leq \liminf_k f(X_{n_k}-s_{n_k})$ and that
\begin{align*}
    \rho^f(X) \leq &f(X-s)-s \leq \liminf_k f(X_{n_k}-s_{n_k})+\lim_k(-s_{n_k})\\
    =&  \liminf_k \Big(f(X_{n_k}-s_{n_k})-s_{n_k}\Big)= \liminf_k \rho^f(X_{n_k}) \leq \lambda. 
\end{align*}
Therefore, $X \in \{\rho \leq \lambda\}$ and $\rho$ is $|| \cdot ||_1$ lower semicontinuous.

\end{proof}

\section{Transformed norm risk measures}
In this section, we will present our application to the study of transformed norm risk measures. We will follow the terminology and notations of \cite{CL:08,CL:09}. We call a function $G:[0,\infty) \rightarrow[0,\infty]$ an \emph{Orlicz function} if it left-continuous, convex, $\lim_{x \rightarrow 0^+}G(x)=G(0)=0$ and $\lim_{x \rightarrow \infty} G(x)=\infty$. Under these assumptions $G$ is increasing (i.e., $x \leq y \implies G(x) \leq G(y)$). The convex conjugate $G^*$ of $G$  is again an Orlicz function. The \emph{Orlicz space} corresponding to $G$ is given by $L^G=\{X \in L^1 \,\, : \,\, \mathbb{E}[G(c|X|)] <\infty \,\, \text{ for some } c>0\}$ and the \emph{Orlicz heart} corresponding to $G$ is given by $M^G=\{X \in L^G \,\, : \,\, \mathbb{E}[G(c|X|)] <\infty \,\, \text{ for any } c>0\} $. The \emph{Luxemburg norm} is given by the following formula $$||X||_{G}=\inf\left\{\lambda>0:\mathbb{E}\left[G\left(\frac{|X|}{\lambda}\right)\right]\leq 1\right\}.$$ 
Recall  that the Orlicz space has the equivalent description $L^G=\{X \in L^1 \,\, : \,\, ||X||_{G}<\infty\}$. Also, note that $G^{-1}(1):=\sup\{t > 0 \,\,: \,\,  G(t)\leq 1\}=\frac{1}{||\one||_G}$. Finally, the \emph{Orlicz norm} is given by the formula $||X||^*_{G}=\sup\{\mathbb{E}[XY] \,\, : \,\, ||Y||_{G} \leq 1\}$.


Let $F$ be a left-continuous, increasing, convex function from $[0,\infty)$ to $(-\infty,\infty]$, not identical equal to $\infty$, such that $\lim_{x \rightarrow \infty} F(x)=\infty$, $G$ a real-valued Orlicz function, and $H: \mathbb{R} \rightarrow [0,\infty)$ an increasing, convex function with $\lim_{x \rightarrow \infty}H(x)=\infty$. We denote with $F^*$ and $H^*$ the convex conjugate of $F$ and $H$ respectively. We will also make use of the following conditions introduced in \cite{CL:08,CL:09}. 
$$(FGH1)  \qquad\quad F(\frac{H(s)+\epsilon}{G^{-1}(1)})<\infty \text{ for some } s \in \mathbb{R} \text{ and } \epsilon>0.$$
$$(FGH2) \qquad\qquad \lim_{s \rightarrow \infty}\Big(F(H(s))-G^{-1}(1)s\Big)=\infty. $$
The \emph{transformed norm risk measure} $T$ is defined as follows
\begin{equation}\label{cash1}
T(X)=\inf_{s \in \mathbb{R}}\{F(||H(s-X)||_G)-s\}. 
 \end{equation}
The risk measure $T$ is well-defined on the following subset of $L^1$:
$$\mathcal{X}=\{X \in L^1 \,\, : \,\, H(s-X) \in L^G  \,\, \text{ for all } s \in \mathbb{R}\}.$$

It is easy to see that $\mathcal{X}$ is convex and $L^\infty \subset \cX$. Moreover, $T:\cX \rightarrow [-\infty,\infty]$ is  monotone, convex and cash-additive. The following example illustrates that $\mathcal{X}$ is not  a vector space in general.

\begin{example}\label{hgexmp}
For $H(x)=x^+$,  $\mathcal{X}=L^G+L^1_+$.
Indeed, let $X \in L^G, Y\in L^1_+$ and $s \in \mathbb{R}$. Then $(s -X)^+ \in L^G$, and since $(s-X-Y)^+ \leq (s  -X)^+ $,   $(s-X-Y)^+ \in L^G$, proving that $L^G+L^1_+ \subset \mathcal{X}$. Now let $X \in \mathcal{X}$. Then $(s-X)^+ \in L^G$, and $X=s-(s-X)^++(s-X)^- \in L^G+L^1_+$. 
\end{example}

In \cite{CL:08,CL:09}, an extensive analysis of the properties of $T$ has been carried out under the assumption that $T$ is restricted on $M^\Phi$, where $\Phi:=G \circ H_0$ and $H_0(x)=H(x)-H(0)$ for $x \geq 0$. We aim here to investigate the properties of $T$ without imposing any restrictions on the domain. To apply the results of the previous section, we will follow \cite{CL:08,CL:09} and view $T$ as a cash-additive hull $\rho^f$ of an appropriate functional $f$.  We define  $\overline{F}:[0,\infty] \rightarrow (-\infty,\infty]$ as follows: \begin{equation}\overline{F}(x)=\begin{cases} F(x), & x \in [0,\infty) \\ \infty, & x=\infty\end{cases}\end{equation} 
and we put
\begin{equation}\label{function_norm}f(X)=\overline{F}(||H(-X)||_G), \quad  X \in L^1.\end{equation}


\begin{lemma}\label{lemmaFGH}
Let $f$ as in ~(\ref{function_norm}). Then
\begin{itemize}
\item[(i)] $f$ is convex, monotone, dilatation monotone, $||\cdot||_1$ lower semicontinuous on $L^1$. 
\item[(ii)] Under (FGH2), we have that 
$\displaystyle\lim_{|s| \rightarrow \infty} f(s )+s=\infty$.
\end{itemize}
\end{lemma}

\begin{proof}
(i) Convexity and monotonicity of $f$ is clear. 
We assert that $f$ is dilatation monotone. 
Let $X \in L^1$ and $\pi \in \Pi$. If $H(-X) \not\in L^G$, then $f(X)=\infty$ and $f(\mathbb{E}[X|\pi])\leq f(X)$. Suppose that $H(-X) \in L^G$. Then by the conditional Jensen's inequality and (\cite{ES:92}, Corollary 2.3.11),  
$$0 \leq H(\mathbb{E}[-X|\pi]) \leq \mathbb{E}[H(-X)|\pi],$$
$$||H(\mathbb{E}[-X|\pi])||_G \leq ||\mathbb{E}[H(-X)|\pi]||_G \leq ||H(-X)||_G.$$
Since $\overline{F}$ is increasing, we conclude that $ f(\mathbb{E}[X|\pi])  \leq f(X)$ as desired.

Next, we show that $f$ is $||\cdot||_1$ lower semicontinuous. Let $\lambda \in \mathbb{R}$ and $(X_n)_{n \in \mathbb{N}} \subset \{f\leq \lambda\}$ be such that  $X_n \xrightarrow{{||\cdot||_1}} X \in L^1$. By passing to a subsequence, we may assume that $X_n \xrightarrow{{a.s.}} X$ and $\sup_{n \in \mathbb{N}}\{|X_n|\}  \in L^1$. Put $Y_n=\sup_{k \geq n}X_k \in L^1$ and note that $Y_n \downarrow X$. Since $H$ is continuous and increasing, $H(-Y_n) \uparrow  H(-X)$. It follows that $||H(-Y_n)||_G \uparrow ||H(-X)||_G$; indeed, if $\sup_{n \in \mathbb{N}}\{||H(-Y_n)||_G\}=\infty$, then the assertion is obvious, otherwise it follows from (\cite{Z:83}, Theorem 131.6). Since $\overline{F}$ is left continuous,   $\lim_n\overline{F}(||H(-Y_n)||_G)=\overline{F}(||H(-X)||_G)$. Now since $f(X_n) \geq f(Y_n)$ for each $n \in \mathbb{N}$, $f(X) =\lim_n f(Y_n) \leq \liminf_n f(X_n) \leq \lambda$, implying that $X \in  \{\rho\leq \lambda\}$. This proves that $f$ is $||\cdot||_1$ lower semicontinuous.



(ii)
We have that $f(s)+s=F(||H(-s)||_G)+s\geq F(0)+s$. Thus clearly, $\lim_{s \rightarrow \infty}f(s)+s=\infty$. We also have that 
$F(||H(-s)||_G)+s=F(||\one||_G H(-s))+s=F(\frac{H(-s)}{G^{-1}(1)})+s $, and consequently,
$$\lim_{s \rightarrow -\infty}f(s )+s=\lim_{s \rightarrow \infty}F(\frac{H(s)}{G^{-1}(1)})-s.$$
Observe that for a convex function $V$ defined on $[0,\infty)$, $\frac{V(s)-V(0)}{s}$ is increasing in $s$ and thus if $V(0) \in \mathbb{R}$ we get $$\lim_{s \rightarrow \infty}\frac{V(s)}{s}=\sup\Big\{\frac{V(s)-V(0)}{s} \,\, : \,\, s \in [0,\infty)\Big\}.
$$
Applying this to $F$ and $H$ and using their monotonicity, we have 
$$\lim_{s \rightarrow \infty}\frac{F(s)}{s}=a, \displaystyle\lim_{s \rightarrow \infty}\frac{H(s)}{s}=b \,\, \text{ for } a,b \in (0,\infty].$$
From this, it follows that for any $c\in(0,\infty)$, $$\lim_{s \rightarrow \infty}\frac{F(\frac{H(s)}{c})c}{s}= \lim_{s \rightarrow \infty}\frac{F(\frac{H(s)}{c})}{\frac{H(s)}{c}} \lim_{s \rightarrow \infty}\frac{\frac{H(s)}{c}}{\frac{s}{c}}=ab.$$
In particular,  $\lim_{s \rightarrow \infty}\frac{F(H(s))}{s}=ab$. We claim that $ab>G^{-1}(1)$. If $F$ takes $\infty$ value, then $a=\infty$ and thus $ab>G^{-1}(1)$. Suppose that $F$ is real-valued and $ab\leq G^{-1}(1)$. Then since $F\circ H$ is convex, we get that $\lim_{s \rightarrow \infty}\frac{F(H(s))}{s}=\sup\{\frac{F(H(s))-F(H(0))}{s} \,\, : \,\, s \geq 0\} \leq G^{-1}(1)$. Hence, $F(H(s))-G^{-1}(1)s \leq F(H(0))$ for each $s \geq 0$, which contradicts (FGH2). This proves the claim. Finally,  we have that 
\begin{align*}
    \lim_{s \rightarrow \infty}F\big(\frac{H(s)}{G^{-1}(1)}\big)-s=&\lim_{s \rightarrow \infty}\frac{s}{G^{-1}(1)}\Big(\frac{F(\frac{H(s)}{G^{-1}(1)})G^{-1}(1)}{s}-G^{-1}(1)\Big)\\
    =&\lim_{s \rightarrow \infty}\frac{s}{G^{-1}(1)}(ab-G^{-1}(1))=\infty.
\end{align*}
This completes the proof.
\end{proof}

Let $\overline{T}$ be the extension of $T$ on $L^1$ given by formula ~(\ref{eq1}). In the following results we show that $\overline{T}$ satisfies the axioms of convex risk measures and admits tractable representations.

\begin{theorem}\label{tnormcor}
Suppose that (FGH2) holds. Then $\overline{T}$ is $||\cdot||_1$ lower semicontinuous, convex, monotone, cash-additive, dilatation monotone and is  given by the following formula
\begin{equation}\label{tnorm}\overline{T}(X)=\inf_{s \in \mathbb{R}}\{\overline{F}(||H(s-X)||_G)-s\}, \quad X \in L^1.\end{equation}
\end{theorem}

\begin{proof}
Let $f$ be as in (\ref{function_norm}). Note that the right hand side of (\ref{tnorm}) is $\rho^f$.  It is clear that  $\rho^f$ extends $T$. Also by Lemma~\ref{lemmaFGH} and Theorem~\ref{cash-thm},  $\rho^f:L^1 \rightarrow (-\infty,\infty]$ is dilatation monotone and $||\cdot||_1$ lower semicontinuous. Therefore  by Proposition~\ref{starproperty} and Theorem~\ref{main_thm2}, we have that $\overline{T}(X)=\rho^f(X)$ for all $X \in L^1$, also $\overline{T}$ preserves the convexity, monotonicity, and cash-additivity of $T$.

\end{proof}

\begin{corollary}\label{tnrom_cor}
Under the assumptions (FGH1) and (FGH2), we have that
\begin{equation}\label{tnorm2}\overline{T}(X)=\sup_{\frac{d\mathbb{Q}}{d\mathbb{P}} \in \mathcal{L}}\Big\{\mathbb{E}_{\mathbb{Q}}[-X]-T^{\#}\big(-\frac{d\mathbb{Q}}{d\mathbb{P}}\big)\Big\} ,\quad\text{ for all } X\in L^1,\end{equation}
where $$T^{\#}(-\frac{d\mathbb{Q}}{d\mathbb{P}})=\min_{\eta \in L^{G^*}_+, \{\eta=0\} \subset \{\frac{d\mathbb{Q}}{d\mathbb{P}}=0\}} \Big\{ \mathbb{E}[\eta H^*(\frac{1}{\eta}\frac{d\mathbb{Q}}{d\mathbb{P}}\one_{\{\frac{d\mathbb{Q}}{d\mathbb{P}}>0\}})]+F^*(||\eta||^*_G)\Big\}.$$
\end{corollary}

\begin{proof}
First we note that (FGH1) ensures that $T$ is not identical equal to $\infty$ (see e.g. (\cite{CL:09}, Lemma 5.2)). By Theorem~\ref{tnormcor} and Theorem~\ref{main_rep_thm}, we have that 
$$\overline{T}(X)=\sup_{\frac{d\mathbb{Q}}{d\mathbb{P}} \in \mathcal{L}}\Big\{\mathbb{E}_{\mathbb{Q}}[-X]-T^{\#}\big(-\frac{d\mathbb{Q}}{d\mathbb{P}}\big)\Big\}, \,\, \quad\text{ for all } X\in L^1,$$
where $T^{\#}$ is given by the formula~(\ref{eq2}), that is $T^{\#}(-\frac{d\mathbb{Q}}{d\mathbb{P}})=\sup_{X \in \mathcal{L}}\{\mathbb{E}_{\mathbb{Q}}[-X]-T(X)\}$. By equation (4.2) in \cite{CL:09} and (\cite{CL:09}, Theorem 5.1) we have that 

$$\sup_{X \in M^{\Phi}}\{\mathbb{E}_{\mathbb{Q}}[-X]-T(X)\}=\min_{\eta \in L^{G^*}_+, \{\eta=0\} \subset \{\frac{d\mathbb{Q}}{d\mathbb{P}}=0\}} \Big\{ \mathbb{E}[\eta H^*(\frac{1}{\eta}\frac{d\mathbb{Q}}{d\mathbb{P}}\one_{\{\frac{d\mathbb{Q}}{d\mathbb{P}}>0\}})]+F^*(||\eta||^*_G)\Big\},
$$

where $\Phi=G \circ H_0$ and $H_0(x)=H(x)-H(0)$, for $x \geq 0$. The result now follows by applying Lemma~\ref{density}(ii) with $\cC=M^{\Phi}$.
\end{proof}

\subsection*{Higher order dual risk measures}
In this subsection, we will discuss the Kusuoka representations of higher order dual risk measures. We recall that the Average Value at Risk of $X \in L^1$ at level $\alpha \in (0,1]$ is given by the following formula
$$\mathrm{AVaR}_{\alpha}(X)=\frac{1}{\alpha} \int_{0}^\alpha \mathrm{VaR}_t(X)\mathrm{d}t,$$
where $\mathrm{VaR}_{t}(X)=\inf\{m \,\, : \,\, \mathbb{P}(X+m <0) \leq t\}$. Also recall that $\mathrm{AVaR}_{\alpha}$ is cash-additive, dilatation monotone and $||\cdot||_1$ continuous.

In the following, fix some $p,c>1$ and put $$G(x)=x^p, \;\;H(x)=x^+,\;\; F(x)=c x.$$ 
Then the extended transformed norm risk measure is given by the following formula.
\begin{equation}\label{Highorder}\overline{T}_{c,p}(X)=\inf_{s \in \mathbb{R}}\{c||(s-X)^+||_p-s\},\qquad X \in L^1,\end{equation}
where $||\cdot||_p$ is the $L^p$-norm. $\overline{T}_{c,p}$ corresponds to the higher order dual risk measure studied in \cite{KPR:10}. In this article, the authors derived a Kusuoka representation of $\overline{T}_{c,p}$ for $X \in L^p$. In the following corollary, we show that the representation holds for any $X \in L^1$.

\begin{corollary}\label{highorder_cor} 
Let $\overline{T}_{c,p}$ be as in (\ref{Highorder}). Then $\overline{T}_{c,p}$ admits the following Kusuoka representation representation. For any $X \in L^1$,
$$\overline{T}_{c,p}(X)=\sup_{\mu \in \mathcal{M}_q} \int_0^1 \mathrm{AVaR}_{\alpha}(X) \mu(\mathrm{d}\alpha),$$
where $q$ such that $\frac{1}{q}+\frac{1}{p}=1$ and $\mathcal{M}_{q}=\{\mu \in \mathcal{P}((0,1]) \,\, : \,\, \int_0^1 | \int_{\alpha}^1 \frac{\mu (\mathrm{d}t)}{t}|^q \mathrm{d}\alpha \leq c^q\}$.
\end{corollary}

\begin{proof}
It is straightforward to see that the condition (FGH2) is satisfied, thus by Theorem~\ref{tnormcor} we have that $\overline{T}_{c,p}$ is dilatation monotone and $||\cdot||_1$ lower semicontinuous. Let $X \in L^1$ and fix $(\pi_n)_{n \in \mathbb{N}} \subset \Pi$ such that $\mathbb{E}[X|\pi_n] \xrightarrow{{||\cdot||_1}} X$. Then by ~(\ref{remark000}),   $\overline{T}_{c,p}(X)=\displaystyle\lim_n \overline{T}_{c,p}(\mathbb{E}[X|\pi_n])$. 

By the Kusuoka representation (\cite{KPR:10}, Theorem 1) and the dilatation monotonicity of $\text{AVaR}_{\alpha}$, we get that $$\overline{T}_{c,p}(\mathbb{E}[X|\pi_n])=\sup_{\mu \in \mathcal{M}_q} \int_0^1 \mathrm{AVaR}_{\alpha}(\mathbb{E}[X|\pi_n]) \mu(\mathrm{d}\alpha)\leq \sup_{\mu \in \mathcal{M}_q} \int_0^1 \mathrm{AVaR}_{\alpha}(X) \mu(\mathrm{d}\alpha) .$$ 
Thus by letting $n \rightarrow \infty$, we get that $$\overline{T}_{c,p}(X) \leq  \sup_{\mu \in \mathcal{M}_q} \int_0^1 \mathrm{AVaR}_{\alpha}(X) \mu(\mathrm{d}\alpha).$$ 
On the other hand, since $\mathrm{AVaR}_{\alpha}$ is $|| \cdot ||_1$ continuous,   $\mathrm{AVaR}_{\alpha}(X)= \lim_n\mathrm{AVaR}_{\alpha}(\mathbb{E}[X|\pi_n])$ for each $\alpha \in (0,1]$. Moreover, for each $n\in\N$, by the dilatation monotonicity of $\mathrm{AVaR}_{\alpha}$ applied to $\E[X|\pi_n]$,   $$ \mathrm{AVaR}_{\alpha}(\mathbb{E}[X|\pi_n])\geq \mathrm{AVaR}_{\alpha}\big(\E\big[\mathbb{E}[X|\pi_n]\big]\big)= \mathrm{AVaR}_{\alpha}(\mathbb{E}[X]) =-\mathbb{E}[X],$$ for each $\alpha \in (0,1]$, where the last equality is due to cash-additivity of $\mathrm{AVaR}_{\alpha}$. Now, fix some $\mu  \in \mathcal{M}_{q}$. By Fatou's Lemma applied to the sequence $\big(\mathrm{AVaR}_{\bullet}(\mathbb{E}[X|\pi_n])\big)_n$,   $$\int_{0}^1 \mathrm{AVaR}_{\alpha}(X) \mu(\mathrm{d}\alpha) \leq \liminf_n \int_{0}^1 \mathrm{AVaR}_{\alpha}(\mathbb{E}[X|\pi_n]) \mu(\mathrm{d}\alpha).$$ By using again the Kusuoka representation (\cite{KPR:10}, Theorem 1), we get that 
\begin{align*}
  \int_{0}^1 \mathrm{AVaR}_{\alpha}(X)\mu(d\alpha) \leq& \liminf_n \sup_{\mu \in \mathcal{M}_q }\int_{0}^1 \mathrm{AVaR}_{\alpha}(\mathbb{E}[X|\pi_n]) \mu(d\alpha) =\liminf_n\overline{T}_{c,p}(\mathbb{E}[X|\pi_n])\\
  =&\overline{T}_{c,p}(X).
\end{align*}
Taking supremum over $\mathcal{M}_q$, we conclude that $$\sup_{\mu \in \mathcal{M}_q} \int_0^1 \text{AVaR}_{\alpha}(X) \mu(d\alpha)\leq \overline{T}_{c,p}(X).$$
This completes the proof.
\end{proof}

\textbf{Acknowledgements.} The authors would like to thank Niushan Gao for valuable comments that improved the exposition of the manuscript.


\end{document}